\documentclass[10pt,a4paper]{article}

\usepackage{amsmath,amssymb,amsthm}
\usepackage{url}
\usepackage{placeins}
\usepackage{subfigure}
\usepackage{multirow}
\usepackage{epsfig}
\usepackage{chicago}
\usepackage{rotating}

\newcommand{\Pbb}{\ensuremath{\mathbb{P}}}

\newcommand{\Qbb}{\ensuremath{\mathbb{Q}}}
\newcommand{\E}{\ensuremath{\mathbb{E}}}

\newcommand{\Abb}{\ensuremath{\mathbb{A}}}

\newtheorem{lemma}{Lemma}

\newcommand\be{$$}
\newcommand\ee{$$}
\newcommand\ben{\begin{equation}}
\newcommand\een{\end{equation}}
\newcommand\bea{\begin{eqnarray*}}
\newcommand\eea{\end{eqnarray*}}
\newcommand\bean{\begin{eqnarray}}
\newcommand\eean{\end{eqnarray}}

\begin{document}
\author{ Chris Kenyon\footnote{Contact: chris.kenyon@lloydsbanking.com}, \ Andrew Green\footnote{Contact: greandrew@googlemail.com}\ \ and \ Mourad Berrahoui\footnote{Contact: mourad.berrahoui@lloydsbanking.com}}
\title{Which measure for PFE?\\The Risk Appetite Measure \Abb\footnote{\bf The views expressed are those of the authors only, no other representation should be attributed.   Not guaranteed fit for any purpose.  Use at your own risk.}}
\date{15 December 2015\vskip5mm Version 1.01}

\maketitle

\begin{abstract}
Potential Future Exposure (PFE) is a standard risk metric for managing business unit counterparty credit risk but there is debate on how it should be calculated.  The debate has been whether to use one of many historical (``physical'') measures (one per calibration setup), or one of many risk-neutral measures (one per numeraire).  However, we argue that limits should be based on the bank's own risk appetite provided that this is consistent with regulatory backtesting and that whichever measure is used it should behave (in a sense made precise) like a historical measure.  Backtesting is only required by regulators for banks with IMM approval but we expect that similar methods are part of limit maintenance generally.  We provide three methods for computing the bank price of risk from readily available business unit data, i.e.\ business unit budgets (rate of return) and limits (e.g.\ exposure percentiles).  Hence we define and propose a Risk Appetite Measure, \Abb, for PFE and suggest that this is uniquely consistent with the bank's Risk Appetite Framework as required by sound governance.
\end{abstract}

\section{Introduction}
Potential Future Exposure (PFE) is a standard risk metric for managing business unit counterparty credit risk \cite{canabarro2003measuring,Pykhtin2007modelling} however there is debate on how it should be calculated \cite{Stein2013a}, as well as regulatory constraints \cite{BCBS-185}.  Whilst backtesting is only required by regulators for banks with IMM approval, we expect that similar methods are part of limit maintenance generally.  The debate has been whether to use one of many historical measures (one per calibration setup), or one of many risk-neutral measures (one per numeraire).  However, we argue that for PFE: 
\begin{itemize}
	\item whatever measure is chosen it should behave (in a sense to be made precise) like a historical measure;
	\item the risk appetite of the bank should be consistent with regulatory constraints;
	\item open risk should be viewed according to the risk appetite of the bank, i.e.\ the bank's price of risk.
\end{itemize}
Given these three constraints we identify a unique new measure, the Risk Appetite Measure, RAM, or \Abb\ for computing PFE.  We demonstrate that this measure can be computed simply from existing bank data and provide examples.

As in \cite{canabarro2003measuring} we define PFE as a quantile $q$ of future exposure, that is, as a VaR measure.  Since we are considering exposure as a control metric this is floored at zero.  It is possible that PFE as VaR may be replaced by PFE as Expected Shortfall (ES) following \cite{BCBS-265,BCBS-325}.  Substituting ES for VaR makes no difference to the development here as both VaR and ES can be backtested \cite{Acerbi2014backtesting} which is sufficient for our purposes.

Open risk, and hence limits, is a widespread feature of banking and one motivation for capital.  Pricing open risk, i.e.\ warehoused risk, is starting to attract attention in mathematical finance \cite{Hull2014joint,Kenyon2015warehousing} and has a long history in portfolio construction \cite{Markowitz1952a} and investment evaluation \cite{sharpe1964capital}.  Valuation adjustments on prices for credit have a long history \cite{Green2015xva} but it is only recently that capital has been incorporated \cite{Green2014kva}.   This paper  adds to the literature linking investment valuation and mathematical finance by proposing that limits be computed under the bank's risk appetite measure provided as part of sound, and self-consistent, governance.

The contribution of this paper is to develop and propose a method of computing PFE that is consistent with a bank's Risk Appetite Framework (RAF) \cite{BCBS-328,FSB2013appetite} using the Risk Appetite Measure \Abb, defined here. This method avoids the issues with \Pbb-measure PFE computation (choice of calibration), and avoids issues with \Qbb-measure PFE computation such as the implicit choice of a price of risk of zero.  This approach can be used whenever there is open risk, both for pricing as well as for limits.  Arguably, computing prices and limits for open positions under \Abb\ is the unique  method that is consistent with the bank's RAF.

\section{Risk Appetite Measure, \Abb}

Here we consider each of the constraints in the introduction and show how this leads to a unique definition of a unique Risk Appetite Measure, \Abb, suitable for PFE.

\subsection{PFE behaviour under \Pbb}

\begin{figure}[htbp]
	\centering
		\includegraphics[width=1.2\textwidth,clip,trim=0 10 0 0]{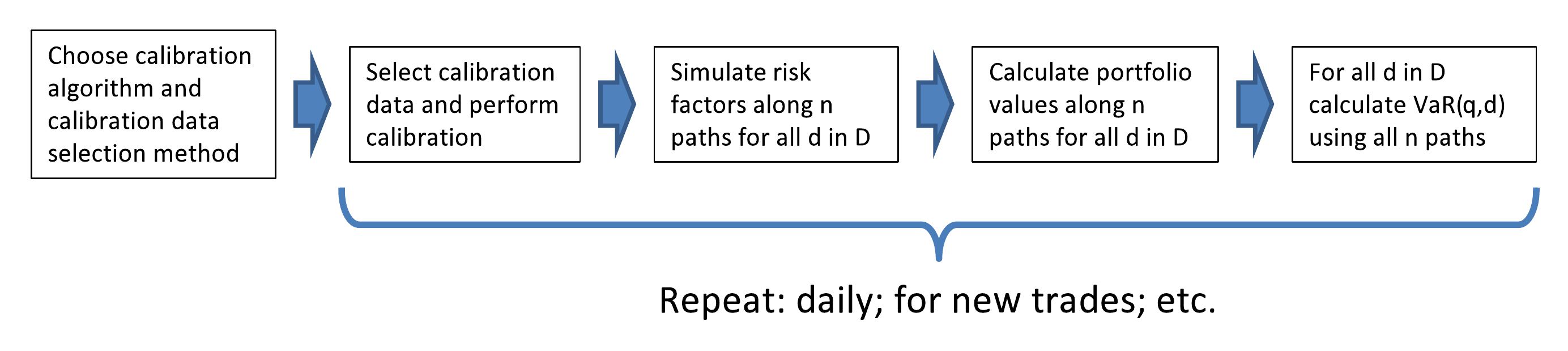}
		\caption{Flow chart for computation of PFE under one of many physical (aka historical) measures.  $d$ is a simulation stopping date in the set of all simulation stopping dates $D$.    Each separate physical measure corresponds to a different choice of calibration algorithm and selection method for calibration data.  We stress that there are many physical measures in practice because of the common, and erroneous,  practice of referring to ``the'' physical measure as though the true future probabilities were available for PFE computation.}
	\label{f:flow}
\end{figure}

Consider the flow chart in Figure \ref{f:flow} for computing PFE where $d$ is an element of the set of simulation stopping dates $D$ and there are $n$ paths.  The output of the procedure is a set of PFE numbers, PFE($q$,$d$), corresponding to the chosen exposure quantile $q$ for each stopping date $d$.

The \Pbb\ procedure of Figure \ref{f:flow} can be characterized by either of the following equivalent properties:
\begin{itemize}
	\item {\it Each} PFE($q$,$d$) computation is independent of {\it all} simulation paths up to $d$
	\item Only portfolio simulation values at $d$ are required for the computation of PFE($q$,$d$) 
\end{itemize}
We call this property of PFE behaviour under \Pbb\ the {\it Independence Property}.  From the point of view of physical-measure simuation this also appear to be common sense, only values at the time have any effect.

Now considering risk-neutral pricing we know from, for example FACT TWO of \cite{Brigo2006a} that the price of a payoff at $T$ seen from $t$
\be
	\text{Price of Payoff(T)}_t = \E_t^{\fbox{B}}\left[  \fbox{B(t) } \frac{\text{Payoff(T)}}{\fbox{B(T)}}    \right] 
	\ee
is invariant under change of numeraire from $B$ to any other valid numeraire.  Valid numeraires are non-zero, self-financing portfolios.  Hence the probability distribution of the payoff at $T$ {\it discounted back to $t$} is also invariant under any change of numeraire as we now prove.

\begin{lemma}
The discounted cumulative distribution function (CDF) of future portfolio value is measure invariant, as is the discounted probability distribution function (PDF) if it exists.
\end{lemma}
\begin{proof}
Consider $\text{dcdf}(t,V)$ the discounted CDF of portfolio value $x(T)$, (i.e.\ gives probability of $x(T)$ being less than or equal to the constant $V$):
\be
\text{dcdf}(t,V) := \E_t^B\left[ B(t) \frac{I(V,x(T))}{B(T)} \right]  
\ee
Where

$x(T)$ is the portfolio value;

$I(V,x(T))$ is an indicator function that is 1 when $V \ge x(T)$ but zero otherwise.

Now $\E[ I(V,x(T)) ]$ is the CDF of $x(T)$, and the $B(t)/B(T)$ is the discounting.  Since dcdf$(t)$ is invariant under change of $B$ for any other numeraire, the discounted CDF is measure invariant.  Hence so is the discounted PDF (if it exists, as PDFs do not have to).
\end{proof}

This suggests a numeraire-independent definition of PFE as discounted PFE.  However, how can we compare \Pbb-measure PFE and discounted-PFE, i.e.\ how do we discount \Pbb-measure PFE back to $t$?  We can be almost sure that the expected value of a \Pbb-bank account will not be worth the inverse of a zero-coupon bond from $t$ to $T$.  (If we could be sure then the \Pbb-measure drift would be the same as the \Qbb-measure drift which cannot be guaranteed).  Thus we can see two methods to discount \Pbb-measure PFE and it is arguable which one to use as both investments are available at $t$ (bank account and zero coupon T-maturity bond).  Should we chose the investment that generates the highest return?  We know the calibration so this is available.  This uncertainty is similar to the uncertainty over the choice of numeraire for \Qbb-measure PFE (just apart from choices of calibration).

The measure-independent definition of risk-neutral PFE is {\it discounted PFE} and this cannot be compared to \Pbb-measure PFE because discounting \Pbb-measure PFE is ambiguous.  That is, \Pbb-measure PFE is naturally in the future whereas \Qbb-measure PFE is naturally in the present (and discounted).

This is where we need the {\it Independence Property} that {\it each} PFE($q$,$d$) computation is independent of {\it all} simulation paths up to $d$.  This selects a unique method to inverse-discount \Qbb-measure PFE from $t$ (now) to $T$ (future), that is  we must use the T-maturity zero coupon bond.  This construction of future \Qbb-measure PFE is measure independent\footnote{Some readers may object that this construction uses a value not at $T$, i.e.\ the start-time values of $T$-maturity zero coupon bonds.  However, these are in the same category as data used in the \Pbb\ construction for calibration: observed data that is not model or simulation dependent and only refers to time $T$.}.  This construction is also consistent with PFE computed using the $T$-Forward numeraire for each $T$.  Thus we have a \Qbb-measure PFE that behaves in the same way as \Pbb-measure PFE, i.e.\ has the {\it Independence Property}.

\subsection{Consistency of Bank Risk Appetite and Regulatory Constraints}

We develop the following points here:  
\begin{itemize}
	\item a bank has a documented risk appetite which defines prices of different risks;
	\item backtesting requirements on risk factor dynamics give limits on drifts (etc.), that is, limits on the bank prices of risks;
	\item we can calculate the bank price of risk from the risk appetite for each business unit using easily available information, i.e.\ its limits and its budget (required rate of return).
\end{itemize}
We suggest that:
\begin{itemize}
	\item the bank price of risk should be consistent with its regulatory requirements as expressed through backtesting limits on the price of risk;
\end{itemize}
and observe
\begin{itemize}
	\item there is no finite risk appetite that is consistent with a zero price of risk, i.e.\ with a risk-neutral-based PFE.
\end{itemize}
In the following section we will develop the Risk Appetite Measure, \Abb, for PFE that is consistent with the bank's risk appetite (as constrained by regulatory requirements).

\subsubsection{Bank Risk Appetite}

Banks take risks, for example lending money, for rewards, that is to make profits.  Banks are required by regulators to document their risk appetite with an effective Risk Appetite Framework (RAF), and  the RAF is a key element of increased supervision endorsed by G20 leaders \cite{FSB2013appetite}.  A Risk Appetite Statement (RAS) is a key part of supervision \cite{BCBS-328}.  The BCBS and the FSB use slightly different terminology but both see a risk appetite statement as vital.  Risk appetite is defined in \cite{BCBS-328} referencing \cite{FSB2013appetite} as:
\begin{quote}
The aggregate level and types of risk a bank is willing to assume, decided in advance and within its risk capacity, to achieve its strategic objectives and
business plan.
\end{quote}
Risk appetite is not a single figure, rather:
\begin{quote}
Risk appetite is generally expressed through both quantitative and qualitative means and should consider extreme conditions, events, and outcomes. In addition, risk appetite should reflect potential impact on earnings, capital, and funding/liquidity.

Senior Supervisors Group, Observations on Developments in Risk Appetite Frameworks and IT Infrastructure, December 23, 2010
\end{quote}
Principle three of the key elements of FSB Principles for effective RAFs is a statement of risk limits which are defined as 
\begin{quote}
Quantitative measures based on forward looking assumptions that allocate the financial institution’s aggregate risk appetite statement (e.g.\ measure of loss or negative events) to business lines, legal entities as relevant, specific risk categories, concentrations, and as appropriate, other levels. 
\end{quote}
Furthermore:
\begin{quote}
The elements of the RAF should be applied at the business line and legal entity levels in a manner that is proportionate to the size of the exposures, complexity and materiality of the risks.
\end{quote}
In practice every business unit has a PFE (or VaR or other) limit {\it and} a budget.  Budget is the word used for the profit over a fixed period, or rate or return on investment (or equity, etc.) that a business unit is required to deliver.  The business unit limits and the business unit budget define the risk appetite of the bank for the risk type that the business unit is engaged in.  These number are readily available within a bank and constantly monitored.  Every business unit head is acutely aware of their progress towards plan on budget and their limit usage.

Now we need to link the price of risk with the risk appetite.  First recall a definition of the {\it market} price of risk following \cite{bjork2004arbitrage}.  Suppose that we have a stock price process $S(t)$ under a physical measure:
\be
dS(t) / S(t) = \mu dt + \sigma dW
\ee
then the market price of risk $m_M$ is defined as:
\ben
m_M := \frac{\mu - r}{\sigma}	\label{e:por}
\een
where $r$ is the riskless rate.  The price of risk gives the additional return required by investors per unit of return volatility $\sigma$.  Of course the return can only be volatile when it is not hedged.  Were the return to be hedged, i.e.\ riskless, then the return is $r$, the riskless rate.
\begin{figure}[htbp]
	\centering
		\includegraphics[width=0.99\textwidth,clip,trim=0 0 0 0]{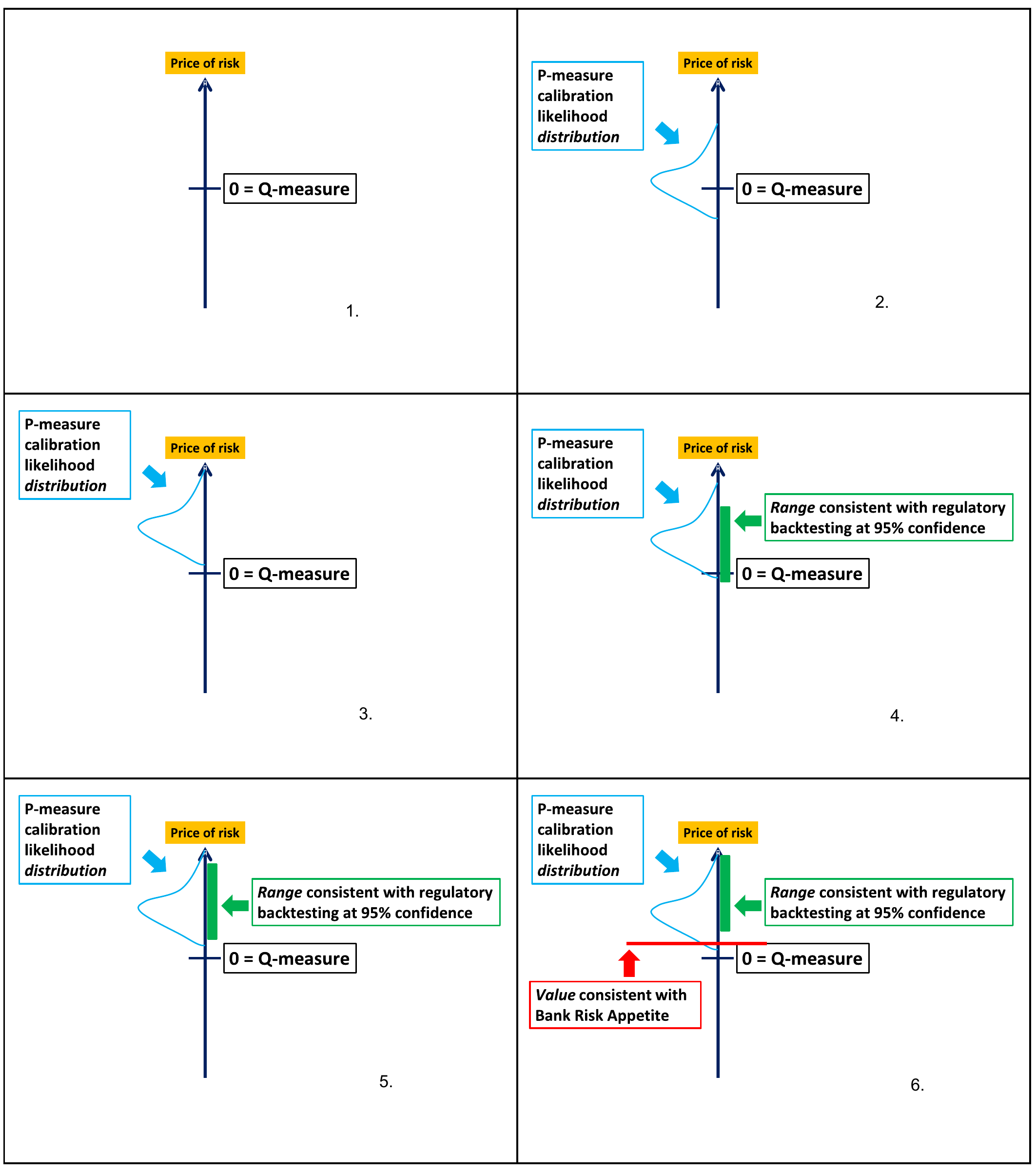}
		\caption{Interaction of \Pbb, \Qbb, and \Abb\ prices of risk.  See text for details.}
	\label{f:intuition}
\end{figure}
Before going into details we give the intuition between the interaction of the different measures and constraints in Figure \ref{f:intuition} as follows.
\begin{enumerate}
	\item (Top Left) shows an axis of possible prices of risk, with zero marked and labelled with the \Qbb\ measure because a risk-neutral measure has a zero price of risk by definition.
	\item (Top Right) adds a probability distribution of \Pbb-measure prices of risk.  A maximum-likelihood calibration will, by definition, pick the price of risk with the highest likelihood.  This particular calibration shows the support of the \Pbb-measure prices of risk including zero.
	\item (Middle Left) considers a later \Pbb-measure calibration where the probability distribution of \Pbb-measure prices of risk has shifted because the input data for the calibration has shifted.  In this example the support of the \Pbb-measure prices of risk no longer  includes zero.
	\item (Middle Right) adds a range of possible prices of risk consistent with regulatory backtesting.  There will be a most likely price of risk from the regulatory backtesting but there will also be an interval (generally, a set) of prices of risk that cannot be ruled out at any given confidence level (here given as 95\%).
	\item (Bottom Left) considers a later situation where the interval of prices of risk that cannot be ruled out with 95\%\ confidence no longer includes zero.
	\item (Bottom Right) adds the price of risk consistent with the Bank Risk Appetite measure \Abb\ (see below for details on computation).  
\end{enumerate}
Figure \ref{f:intuition} illustrates several potential issues with \Pbb, \Qbb, and \Abb\ prices of risk.  Firstly, there is no guarantee that a zero price of risk will be consistent with a \Pbb-measure price of risk, even if it was previously.  Secondly, even if a zero price of risk is consistent with regulatory backtesting over one period there is no guarantee that a zero price of risk will remain consistent as markets move.  Thirdly, if the Bank's Risk Appetite is chosen without consideration of the regulatory backtesting there is no guarantee that the Bank's Risk Appetite will be consistent with regulatory backtesting.  If this were to occur the Bank might wish to reconsider its risk appetite.

\subsubsection{Limits from Backtesting}

Basel III \cite{BCBS-189} and its European implementation \cite{CRD-IV-Regulation,CRD-IV-Directive} mandate sound backtesting for the internal model method (IMM) for credit exposure.  Principles for such backtesting are provided in \cite{BCBS-185} and expanded on in \cite{Kenyon2012a,Anfuso2014a}.  If a bank is not using internal models for credit exposure for regulatory purposes, i.e.\ does not have regulatory approval to do so, then the bank may or may not choose to follow such principles for other purposes such as PFE.  We would expect this to be part of limit maintenance generally.  However, in the reverse case where regulatory approval has been obtained, then the Use Test element \cite{BCBS-128}, paragraphs 49-54 states that the bank should be consistent in the data used for day-to-day counterparty credit risk management.  Although backtesting is only required by regulators for banks with IMM approval we expect similar methods to be part of limit maintenance generally.

Backtesting will select \Pbb\ calibrations that pass at a desired confidence level:  assuming such exist, if not the bank will not be permitted to use the IMM approach for regulatory purposes, but may still use it for other purposes.  Hence it is possible to adapt the backtesting machinery to provide a set, or interval, or simulation parameters that pass at a given confidence level as in standard statistical interval estimation, see Chapter 9 of \cite{casella2002statistical}.

We suggest that whatever the bank chooses as its price of risk, this should be consistent with historical backtesting.  The regulatory Use Test may in fact require this but there is no specific mention in any of the FAQs for Basel III.  

\subsubsection{Computing the Bank Price of Risk}

Since a bank has a risk appetite it has {\it already} specified its prices of risk.  Here we demonstrate how to compute them from business unit limits and business unit budget (i.e.\ rate of return on investment).

A bank can have different prices of risk for different business units and for different types of risk.  There is a vast literature on market prices for different types of risk including early work by \cite{sharpe1964capital,stanton1997nonparametric,bondarenko2014variance}, here we deal with computing bank prices of risk for use in PFE. We use the same form as Equation \ref{e:por} to compute the Bank price of risk $m_B$, e.g.
\be
m_B = \frac{(\text{Business Unit Rate of Return on Investment}) - r}{\sigma_\Rightarrow}
\ee
where: $r$ is riskless rate; $\sigma_\Rightarrow$ is the implied volatility from the PFE (or VaR) limit and required rate of return for the business unit  (i.e.\ budget versus investment).  Economically profits can come from many sources, for example rent from a monopoly position, or risk taking.  Here we consider only that part of profit that comes from risk taking.

We will start by assuming different parametric forms for the distribution of the rate of return of the business unit.  After this will provide an extension that requires only an empirical (non-parametric) distribution of the rate of return of the business unit.

\paragraph{Normal relative returns}  Suppose the relative returns of the business unit have a Normal($\mu_{\text{required}},\sigma$) distribution. $\mu_{\text{required}}$ is the required rate of return of the business unit, and the business unit limit $L$ is expressed as a PFE($q$) in units of required relative rate of return.  Now we have two unknowns and two constraints so the unknowns are uniquely defined as:
\begin{align}
\sigma_{\Rightarrow} =& \frac{\mu_{\text{required}}(1-L)}{\sqrt{2}\ \text{erfc}^{-1}(2q)}   \label{e:vol} \\
m_B =& \frac{\mu_{\text{required}}-r}{\sigma_\Rightarrow}
\end{align}
$\text{erfc}^{-1}$ is the inverse of the cumulative error function.
\begin{figure}[htbp]
	\centering
		\includegraphics[width=0.70\textwidth,clip,trim=0 0 0 37]{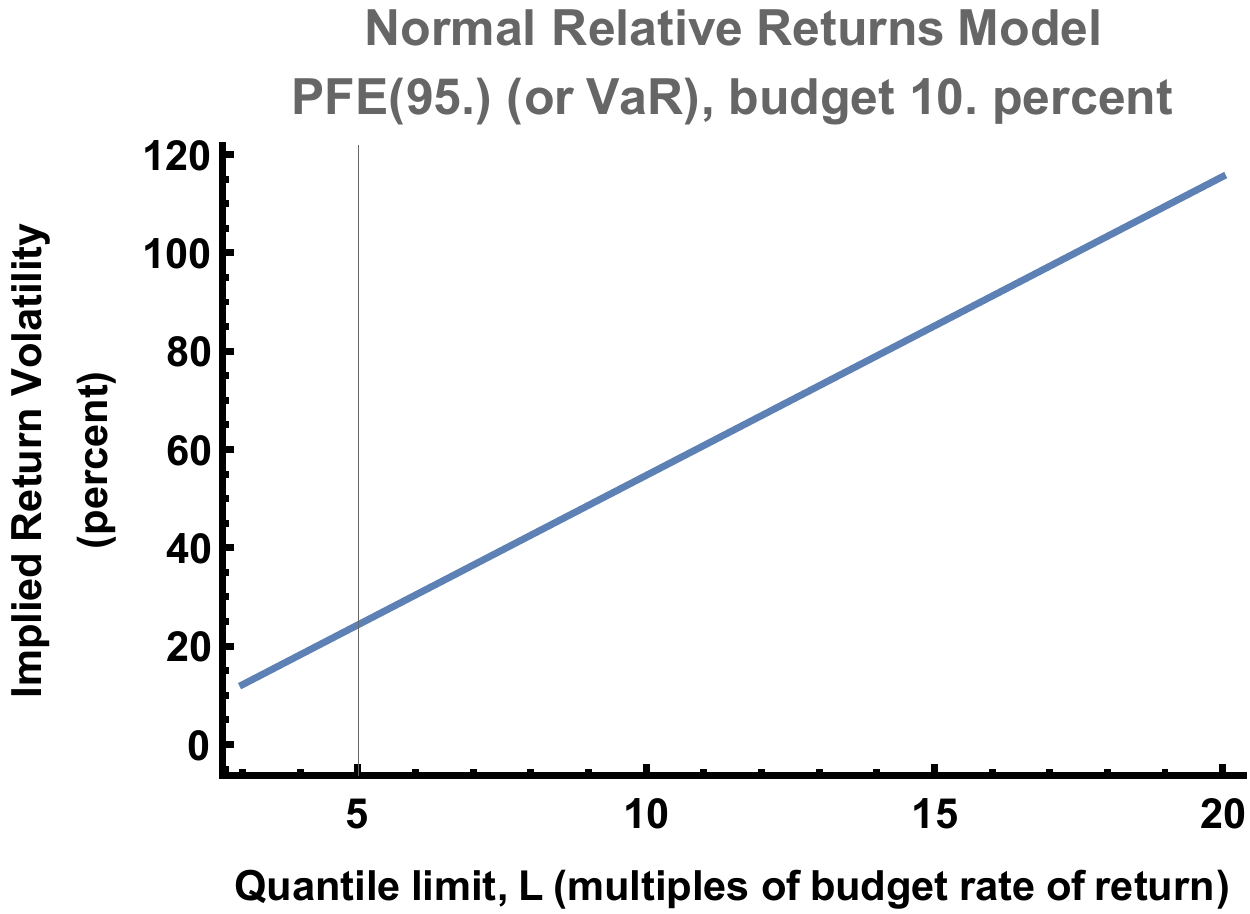}
		\caption{Implied volatility of required rate of return given budget (required rate or return) of 10\%\  and quantile (PFE or VaR) limit using 95\%.    See text for details.}
	\label{f:normal:vol}
\end{figure}
\begin{figure}[htbp]
	\centering
		\includegraphics[width=0.70\textwidth,clip=true,trim=0 0 0 37]{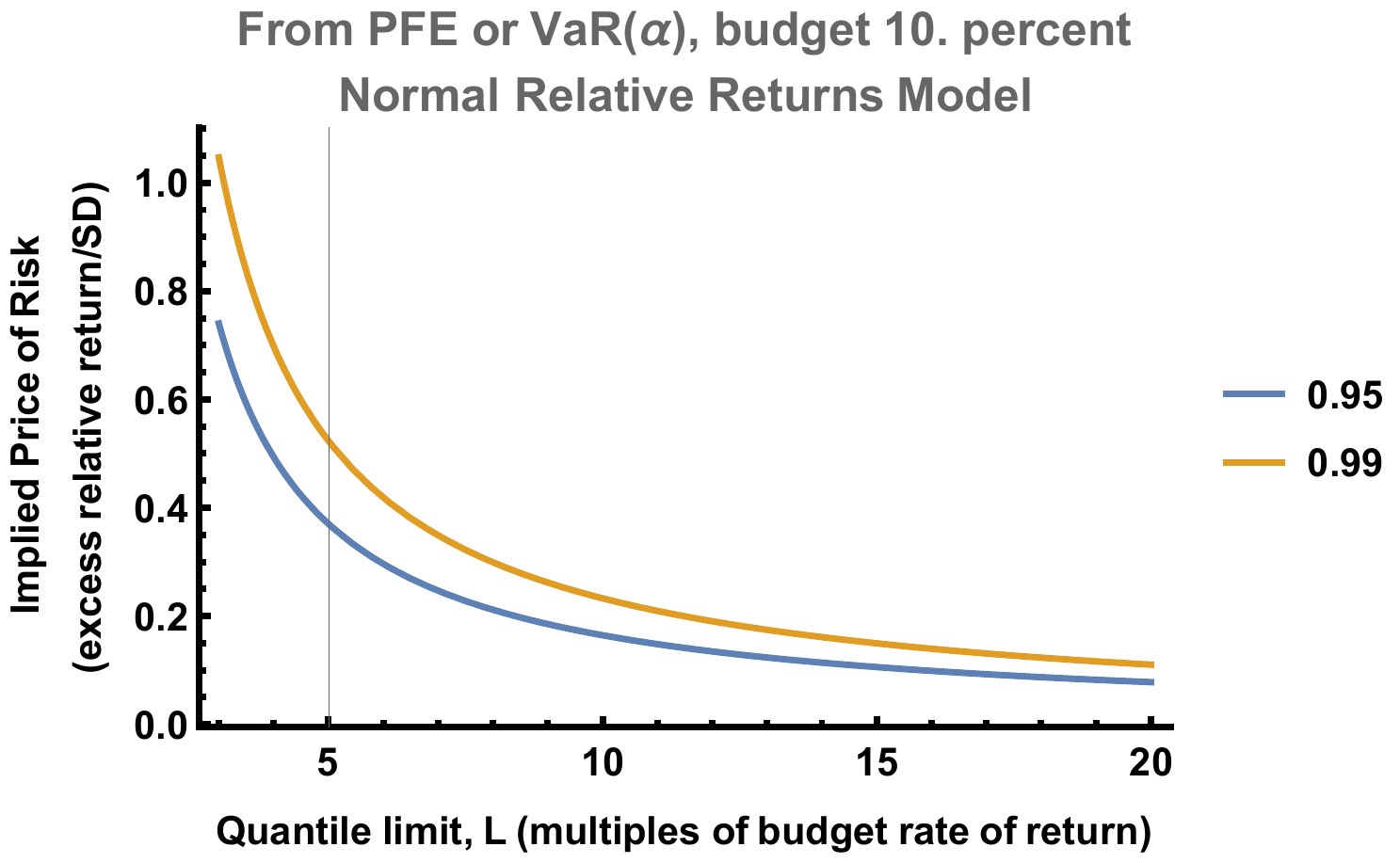}
		\caption{Implied bank price of risk given budget (required rate or return) of 10\%\  and quantile (PFE or VaR) limit using 95\% and 99\%.    See text for details.}
	\label{f:normal:por}
\end{figure}
Figure \ref{f:normal:vol} shows implied volatility of required rate of return given budget (required rate or return) of 10\%\  and PFE (or VaR) limit using 95\%.  As in Equation \ref{e:vol} this implied volatility is linear in the PFE (or VaR) limit expressed in terms of the business unit budget (not percentile).  Figure  \ref{f:normal:por} shows implied bank price of risk given budget (required rate or return) of 10\%\  and PFE (or VaR) limit using 95\% and 99\%.  As the PFE (or VaR) limits increase we see that the implied price of risk per extra unit of volatility decreases as we would expect.  That is, wider limits imply lower aversion to risk. 

\paragraph{Shifted-Lognormal relative returns} To assess the robustness of the prices of risk in Figure  \ref{f:normal:por} we now assume that the relative returns of the business unit have a shifted Lognormal distribution.  This provides a heavy-tailed distribution.  We pick $-1$ as the shift as this is essentially the worst case which we assume is the business unit losing its investment.  Since a shifted Lognormal distribution has three parameters and we have fixed one of them we can solve for the other two parameters.  This involves a quadratic equation so there are two possible solutions.  We chose the lower of the two as the more reasonable, i.e.\ given a low and a high implied price of risk it is more reasonable to conclude that the bank as the lower price of risk.  Of course, depending on circumstances (e.g.\ observations of limits on other business units with Normal returns) it could be reasonable to make the opposite conclusion.  Hence:
\begin{align}
\sigma_{\text{SLN}} =& -\sqrt{2} \text{erfc}^{-1}(2 \text{q}) \nonumber \\
&{}\pm \sqrt{  2 \text{erfc}^{-1}(2 \text{q})^2-2 (\log (  L \mu_{\text{required}} -\gamma ) - \log (\mu_{\text{required}}-\gamma) ) }  \\
\mu_{\text{SLN}}  =& \log (\mu_{\text{required}}-\gamma )-\frac{(\sigma_{\text{SLN}})^2}{2}\\
\sigma_{\Rightarrow}  =&  \sqrt{\left(e^{\sigma_{\text{SLN}} ^2}-1\right) e^{2 \mu_{\text{SLN}}+\sigma_{\text{SLN}}^2}}\\
m_B =& \frac{\mu_{\text{required}}-r}{\sigma_\Rightarrow}
\end{align}
We use $\gamma$ for the shift, giving the general result for a shifted Lognormal.
\begin{figure}[htbp]
	\centering
		\includegraphics[width=0.95\textwidth,clip,trim=0 0 0 37]{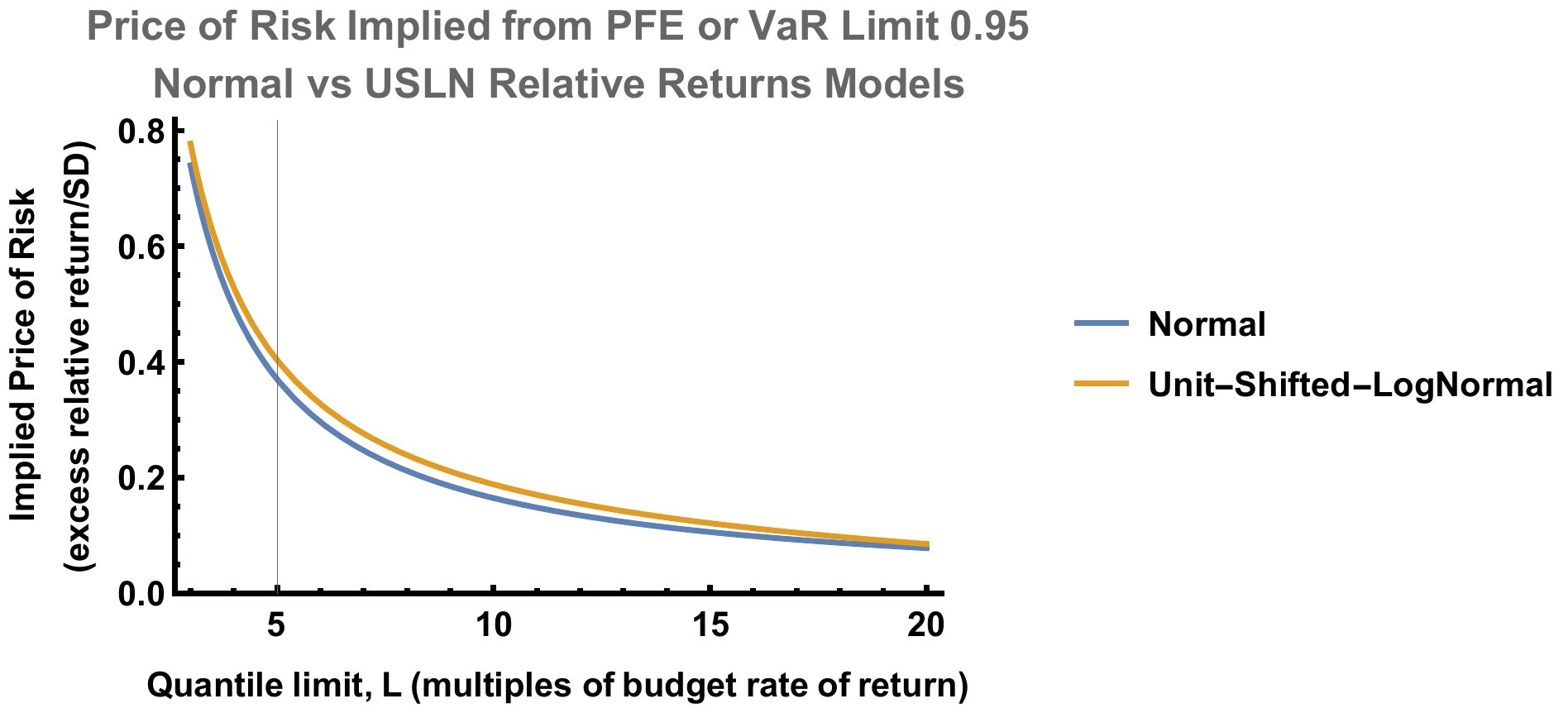}
		\caption{Comparison of implied bank price of risk using two different relative return distributions: Normal and Unit-Shifted-Lognormal}
	\label{f:robust}
\end{figure}
Figure \ref{f:robust} shows a comparison of implied bank price of risk using two different relative return distributions.  These have a maximum relative  difference of 15\% at $L=12.5$, thus we conclude that the implied bank price of risk is robust against changes in the return specification between Normal and Unit-Shifted-Lognormal.

It is clear from Figure \ref{f:robust}, and from simple observation, that no finite PFE (or VaR or ES) limit is compatible with a zero price of risk.  By definition a zero price of risk indicates indifference to anything except the expected return.  This implies that using any risk-neutral measure alone is incompatible with a finite risk appetite.

\paragraph{Empirical relative returns}  There are two cases here: empirical distribution fits the budget and PFE (or VaR) constraints; empirical distribution does not fit the constraints.  If the empirical distribution fits the constraints all that is required is to compute the standard deviation of the relative returns which is trivial.  If the empirical distribution does not fit then we shift and scale it so that it does fit as follows.

Let the {\it ordered}, from smallest to largest, relative returns be $x_i,\ i=1,\ldots,n$ and
\begin{align}
n_q =& \lceil n\times q \rceil  \\
\mu_{\text{required}} =&  g\sum_{i=1}^n ( x_i + h)\\
\text{PFE}(q) = L\mu_{\text{required}} =& g (x_{n_q} + h)
\end{align}
So again we have two equations in two unknowns and can solve for the shift $h$ and scale $g$. Hence
\begin{align}
h =& \frac{x_{n_q} - L\sum_{i=1}^n x_i}{L n  -1}\\
g =& \frac{L\mu}{x_{n_q}+h}
\end{align}
The problem is now reduced to the empirical distribution case which fits the constraints.  That is we simply calculate the empirical standard deviation for $\sigma_{\Rightarrow}$.  This can then be used to give the bank price of risk.

Whether or not it makes sense to scale and shift the empirical distribution will depend on the individual circumstances.  Here we provide the procedure in case it is reasonable.

\subsection{The Risk Appetite Measure, \Abb\ for PFE}

Just as a \Pbb-measure is defined as a change in the drift from a risk-neutral measure given by the market price of risk, we define the Risk Appetite Measure, \Abb, as a change in the drift from a risk-neutral measure given by the bank price of risk, i.e.
\be
\mu_\Abb := m_B \sigma + r
\ee
This is used in conjunction with the T-Forward measure to provide a unique Risk Appetite Measure for PFE.  Specifically we compute the discounted PFE (or VaR, or other limit), which is measure-independent, and then apply an inverse discount given by zero coupon bond prices {\it assuming that the rate of return is 
$\mu_\Abb$}.  This procedure explicitly obeys the {\it Independence Principle} described above and implicitly makes use of the T-Forward measure.  However, note that the procedure permits computation of the discounted PFE (or VaR) under {\it any} measure.  Thus there is no added complexity for computation.

\section{Discussion and Conclusions}

Here we propose using the bank's own risk appetite to set the measure for risk limits, e.g.\ PFE, provided that the bank's risk appetite is consistent with regulatory, or general limit maintenance, backtesting.  If the bank's risk appetite is not consistent with regulatory backtesting this suggests that the bank's risk appetite should be modified in order to pass the regulatory Use Test for day-to-day business decisions and satisfy the principles of a consistent Risk Appetite Framework \cite{BCBS-328,FSB2013appetite}.  We note that a zero price of risk is inconsistent with a finite risk appetite, and this suggests that measures for PFE cannot be risk neutral.

We call this measure defined by the bank's own risk appetite the Bank Risk Appetite Measure or \Abb.  We provide a measure-independent procedure for computing PFE under \Abb.  \Abb\ is defined by the bank price of risk and we give three methods for computing it based on business unit budgets and limits.  This data is readily available and business unit heads are acutely aware of their progress against plan on budget and limit usage.

This paper  adds to the literature by further linking investment valuation and evaluation  \cite{Markowitz1952a,sharpe1964capital} and mathematical finance by proposing that limits be computed under the bank's risk appetite measure provided as part of sound, and self-consistent, governance.

\section*{Acknowledgements}

The authors would like to acknowledge useful discussions with Chris Dennis and Manlio Trovato, and feedback from participants at RiskMinds International (Amsterdam, December 2015) where an earlier version of this work was presented.

\bibliographystyle{chicago}
\bibliography{kenyon_general}

\end{document}